\documentclass[letterpaper, 10 pt, conference, final]{packages/ieeeconf}

\usepackage{times}

\usepackage{amsmath}
\allowdisplaybreaks
\usepackage{amssymb}
\usepackage{bm}

\IEEEoverridecommandlockouts
\usepackage[caption=false]{subfig}

\usepackage{graphicx}
\usepackage{epsfig}
\graphicspath{{plots/}}
\usepackage{pgfplots}
\pgfplotsset{compat=newest}
\usepackage{tikz}
\usepackage{tikzscale}

\usepackage{xfrac}
\usepackage[per-mode=symbol,binary-units=true,group-separator={,}]{siunitx}
\usepackage{csquotes}
\usepackage{physics}
\usepackage[short,c2,nocomma]{packages/optidef}
\usepackage{algorithm}
\usepackage[noend]{algpseudocode}

\makeatletter \let\NAT@parse\undefined \makeatother
\usepackage[colorlinks=true,allcolors=black]{hyperref}

\newlength{\figureheight}
\newlength{\figurewidth}

\DeclareMathAlphabet{\mymathbb}{U}{bbold}{m}{n}
\def \Reals {\mathbb{R}}

\def \Trans {\top}

\DeclareMathOperator{\Unif}{Unif}

\DeclareMathOperator{\Gauss}{\mathcal{N}}

\newtheorem{theorem}{Theorem}
\newtheorem{lemma}{Lemma}
\newtheorem{corollary}{Corollary}
\newtheorem{assumption}{Assumption}

\newtheorem{example}{Example}

\newcommand\copyrighttext{%
	\footnotesize \textbf{Published in: 2019 European Control Conference (ECC), DOI: \href{https://doi.org/10.23919/ECC.2019.8795823}{10.23919/ECC.2019.8795823}.}\\
	\textcopyright 2019 IEEE. Personal use of this material is permitted. Permission from IEEE must be obtained for all other uses, in any current or future media, including reprinting/republishing this material for advertising or promotional purposes, creating new collective works, for resale or redistribution to servers or lists, or reuse of any copyrighted component of this work in other works.}
\newcommand\copyrightnotice{%
	\begin{tikzpicture}[remember picture,overlay]
	\node[anchor=south,yshift=3pt] at (current page.south) {\fbox{\parbox{\dimexpr\textwidth-\fboxsep-\fboxrule\relax}{\copyrighttext}}};
	\end{tikzpicture}%
}

\begin{document}

\title{\LARGE \bf
	Using Uncertainty Data in Chance-Constrained Trajectory Planning
}

\author{Vasileios Lefkopoulos and Maryam Kamgarpour
\thanks{The work of M.\ Kamgarpour is gratefully supported by Swiss National Science Foundation, under the grant SNSF 200021\_172782.}%
\thanks{Both authors are with the Automatic Control Laboratory, ETH Zurich, Z{\"u}rich, CH-8092, Switzerland \{\href{mailto:vlefkopo@student.ethz.ch}{{\tt vlefkopo@student.ethz.ch}}, \href{mailto:mkamgar@control.ee.ethz.ch}{\tt mkamgar@control.ee.ethz.ch}\}.}%
}

\maketitle
\copyrightnotice
\thispagestyle{empty}
\pagestyle{empty}

\begin{abstract}
	We consider the problem of trajectory planning in an environment comprised of a set of obstacles with uncertain locations. While previous approaches model the uncertainties with a prescribed Gaussian distribution, we consider the realistic case in which the distribution's moments are unknown and are learned online. We derive tight concentration bounds on the error of the estimated moments. These bounds are then used to derive a tractable and tight mixed-integer convex reformulation of the trajectory planning problem, assuming linear dynamics and polyhedral constraints. The solution of the resulting optimization program is a feasible solution for the original problem with high confidence. We illustrate the approach with a case study from autonomous driving.
\end{abstract}

\section{INTRODUCTION} \label{sec:Introduction}
Trajectory planning in uncertain environments arises in several autonomous system applications. One method to incorporate the environment's uncertainty into an optimization framework is through chance-constrained programming. With this approach, an acceptable risk level is prescribed and trajectories are constrained to reside in a corresponding probabilistically \enquote{safe} set. The advantage of this approach, in contrast to robust optimization, is that allowing a small level of constraint violation leads to less conservative trajectories and, moreover, can also handle uncertainties with unknown or unbounded support sets.

Mixed-integer programming is often used for dealing with the non-convex nature of obstacle avoidance, as was done for example in~\cite{Schouwenaars2001,Richards2002}. In conjunction with this, Boole's inequality is commonly employed in order to decouple joint chance constraints into single ones, e.g.\ in~\cite{Blackmore2011,Jha2018}. These two methods have been successfully used in recent literature, thus also making them a natural choice for our work.

For tractability, the uncertainties in trajectory planning problems are often assumed to follow Gaussian distributions. This allows reformulations of the chance constraints based on the distributions' moments. In~\cite{Vitus2011} the authors considered planning in convex regions under Gaussian plant uncertainties. Chance-constrained trajectory generation in non-convex regions was tackled in~\cite{Blackmore2011}. In~\cite{Jha2018}, the uncertainties originated from the environment as uncertain polyhedrons. Both system and environment uncertainties were considered in~\cite{Vitus2012} and handled via approximations. For arbitrary distributions, the Cantelli inequality was used in~\cite{Paulson2017} to tackle a chance-constrained model predictive control problem. In all of the above, the standing assumption was that the moments of the uncertainty were \emph{known a priori}. In several realistic scenarios, the autonomous system must maneuver in a partially known environment and can only learn about environment uncertainties online based on real-time sensor measurements.

One method to relax the assumption of known moments is to use samples of the distribution to estimate these moments and to derive probabilistic bounds on the error of these estimates through moment concentration inequalities. Distribution-free concentration inequalities derive such bounds for the case in which no assumption is made about the uncertainty's distribution. The reformulation of chance constraints using estimates of moments has been addressed for distribution-free cases in~\cite{Calafiore2006,Delage2010}. If no assumptions are made on the distribution, the concentration inequalities derived are naturally conservative. As a result, this approach often results in infeasibility for trajectory planning problems. 

Another recent approach to solving chance-constrained trajectory generation problems has been through utilizing collected data directly. In~\cite{Sessa2018} the scenario approach was used to solve a chance-constrained problem arising from deriving robust policies for trajectory planning. The scenario approach to this problem requires a large number of samples to provide high confidence feasibility guarantees. On the one hand, such high number of samples may not be available. On the other hand, the resulting optimization problem is too large for real-time implementation.

We attempt to reach a compromise between the purely data-driven approach of~\cite{Sessa2018} and the purely known distribution assumptions of~\cite{Jha2018} for trajectory generation in stochastic environments. In particular, we take on the fairly accepted model of a Gaussian distribution but consider the realistic case in which the moments of the distribution are \emph{unknown} and are estimated through samples of the uncertainties. Such samples could be obtained through robots sensors online. While the problem of chance-constrained programming under uncertain moments has been addressed through the seminal work of~\cite{Calafiore2006}, the case of tight concentration inequalities given distribution model and using this in a mixed-integer setting to provide probabilistic guarantees for stochastic trajectory generation has not been addressed. To this end, we derive novel and tight concentration bounds on the estimation error of these moments, using the knowledge of the Gaussian distributions. Based on these bounds, we propose a reformulation of the chance-constrained problem using only the estimates of the Gaussian moments and we prove its solution's feasibility (in a probabilistic sense) for the original problem. The provided proof holds regardless of the number of samples available. We demonstrate the performance of this method through a case study in autonomous driving.

\subsubsection*{Notation}
We denote a conjunction (logical AND operator) by $\bigwedge$ and a disjunction (logical OR operator) by $\bigvee$. By $\Gauss(\mu,\Sigma)$ we denote the Gaussian distribution with mean vector $\mu$ and covariance matrix $\Sigma$. By $\Psi^{-1}(\cdot)$ we denote the inverse cumulative distribution function of $\Gauss(0,1)$. We denote a positive definite matrix by $A \succ 0$. Given a random variable $d$, we call its samples $d_i$ as i.i.d.\ if they are independent and identically distributed. By $\Unif(a,b)$ we denote the continuous uniform distribution over $[a,b]$.

\section{PROBLEM STATEMENT} \label{sec:ProblemStatement}

\subsection{Chance-constrained trajectory planning} 
We consider the discrete-time system:
\begin{equation} \label{eq:x_Dynamics}
	x_{t+1} = A_t x_t + B_t u_t \,,
\end{equation}
where $x_t \in \Reals^{n_x}$ is the state and $u_t \in \Reals^{n_u}$ is the input at time $t$. Given an initial state $x_0$ and a horizon length $N$ our goal is to choose the input sequence $\vb{u} \coloneqq (u_0,\dots,u_{N-1})$, with $u_t \in \mathcal{U}$, leading to the state trajectory $\vb{x} \coloneqq (x_1,\dots,x_N)$, such that a given cost $J(x_0,\vb{u})$ is minimized. For tractability, we assume that $\mathcal{U}$ is a compact polyhedron and that the cost $J(x_0,\vb{u})$ is convex with respect to its arguments.

We also require that each state $x_t$ belongs to a \enquote{safe} set $\mathcal{X}_t$, $t=1,\dots,N$. The set $\mathcal{X}_t$ consists of the area outside of polyhedral obstacles with \emph{uncertain faces}. Let $N_o$ be the number of obstacles, with each obstacle indexed by $j$, and let $F_j$ be the number of faces of the $j$-th obstacle, with each face indexed by $i$. Then, the set $\mathcal{X}_t$ can be written as:
\begin{equation} \label{eq:Xt_SafeSet}
	\mathcal{X}_t \coloneqq \left\{ x_t \in \Reals^{n_x} : \bigwedge_{j=1}^{N_o} \bigvee_{i=1}^{F_j} {a^t_{ij}}^\Trans x_t + b^t_{ij} > 0 \right\} \,,
\end{equation}
where $a^t_{ij} \in \Reals^{n_x}$ and $b^t_{ij} \in \Reals$ are the uncertain coefficients of the $i$-th face of the $j$-th obstacle at time $t$. The coefficients $a^t_{ij}$ and $b^t_{ij}$ are modelled as random variables, alternatively concatenated through $d \coloneqq [ {a^t_{ij}}^\Trans \,\, b^t_{ij} ]^\Trans \in \Reals^{n_x+1}$. Due to the stochasticity in $d^t_{ij}$, we enforce the safety constraints through a \emph{joint} chance constraint:
\begin{equation} \label{eq:xt_ChanceConstraint}
	\Pr(\bigwedge_{t=1}^{N} x_t \in \mathcal{X}_t) \geq 1-\epsilon \,,
\end{equation}
where $\epsilon \in (0,0.5)$ is a prescribed safety margin. By combining \eqref{eq:Xt_SafeSet} and \eqref{eq:xt_ChanceConstraint} we obtain the set describing the constraint on the trajectory $\vb{x} \in \Reals^{Nn_x}$:
\begin{equation} \label{x_JointChanceConstraint}
	\bm{\mathcal{X}} \coloneqq \left\{ \vb{x} \in \Reals^{N n_x} : \bigwedge_{t=1}^{N} \bigwedge_{j=1}^{N_o} \bigvee_{i=1}^{F_j} {a^t_{ij}}^\Trans x_t + b^t_{ij} > 0 \right\} \,.
\end{equation}

The optimization problem can then be formulated as:
\vspace{-1.3\baselineskip}
\begin{mini!}<b>
	{\vb{u}}{J(x_0,\vb{u}) \label{eq:CCFHOC_cost}}
	{\label{eq:CCFHOC}}{}
	\addConstraint{\vb{x},\vb{u} \text{ satisfy \eqref{eq:x_Dynamics} with initial state } x_0}{\label{eq:CCFHOC_dynamics}}
	\addConstraint{\vb{u} \in \bm{\mathcal{U}}} {\label{eq:CCFHOC_inputs}}
	\addConstraint{\Pr(\vb{x} \in \bm{\mathcal{X}}) \geq 1-\epsilon}{\label{eq:CCFHOC_constr}}
\end{mini!}
\vspace{-\baselineskip}\par\noindent
where $\bm{\mathcal{U}} \coloneqq \mathcal{U} \cross \dots \cross \mathcal{U} \subseteq \Reals^{N n_u}$.

\subsection{Reformulation as second-order cone constraints} \label{sec:Reformulation}
Problem \eqref{eq:CCFHOC} is not in a form that can be directly handled by off-the-shelf optimization solvers because of the chance constraint \eqref{eq:CCFHOC_constr}. As such, we reformulate this constraint into a tractable form based on the so-called Big-M reformulation~\cite{Schouwenaars2001} and Boole's inequality~\cite{Casella2002}. The latter allows us to split \eqref{eq:CCFHOC_inputs} into \emph{single} chance constraints which are easier to handle. Another alternative for dealing with such constraints could have been e.g.\ the scenario approach as in~\cite{Sessa2018}.

\subsubsection{Big-M Method}
It is known that satisfaction of a disjunction as in \eqref{eq:Xt_SafeSet} can be equivalently written as~\cite{Schouwenaars2001}:
\begin{equation} \label{eq:CCFHOC_Constr_BigM}
	\bigvee_{i=1}^{F_j} {a^t_{ij}}^\Trans x_t + b^t_{ij} > 0 \Leftrightarrow \bigwedge_{i=1}^{F_j} {a^t_{ij}}^\Trans x_t + b^t_{ij} + Mz^t_{ij} > 0 \,,
\end{equation}
where $z^t_{ij} \in \{0,1\}$ are binary variables, for which it must hold that at least one of the $i$-indexed $z^t_{ij}$ is not equal to 1, and $M$ is a positive number that is larger than any possible constraint value in the problem being considered (Big-M). Introducing these binary variables and using \eqref{eq:CCFHOC_Constr_BigM}, constraint \eqref{eq:CCFHOC_constr} can be equivalently rewritten as:
\begin{equation} \label{JointChanceConstraint_BigM}
	\Pr(\bigwedge_{t=1}^{N} \bigwedge_{j=1}^{N_o} \bigwedge_{i=1}^{F_j} {a^t_{ij}}^\Trans x_t + b^t_{ij} + Mz^t_{ij} > 0) \geq 1-\epsilon \,,
\end{equation}
with the additional constraint $\sum_{i=1}^{F_j} z^t_{ij} < F_j$ for all $t$ and $j$, making Problem \eqref{eq:CCFHOC} a mixed-integer optimization problem.

\subsubsection{Boole's Inequality}
Boole's inequality states that for a finite number of events $A_i$ we have $ \Pr(\bigvee_i A_i) \leq \sum_i \Pr(A_i) $~\cite{Casella2002}. Using this the constraint \eqref{JointChanceConstraint_BigM} can be conservatively satisfied through~\cite{Jha2018}:
\begin{equation} \label{CCFHOC_SingleCCs}
	\bigwedge_{t=1}^{N} \bigwedge_{j=1}^{N_o} \bigwedge_{i=1}^{F_j} \Pr({a^t_{ij}}^\Trans x_t + b^t_{ij} + Mz^t_{ij} > 0) \geq 1-\epsilon^t_{ij} \,,
\end{equation}
where $0 \leq \epsilon^t_{ij} \leq 1$ are risk variables that must satisfy $ \sum_{t=1}^{N} \sum_{j=1}^{N_o} \sum_{i=1}^{F_j} \epsilon^t_{ij} \leq \epsilon $. Here, for simplicity, we use the uniform risk allocation $\epsilon^t_{ij} = \flatfrac{\epsilon}{(N \sum_{j=1}^{N_o} F_j)}$. Other existing approaches based on heuristics \cite{Jha2018} or optimization problems \cite{Vitus2011} can easily be incorporated.

To reformulate \eqref{CCFHOC_SingleCCs} into deterministic constraints on the moments of $d^t_{ij}$ we make the following assumption.
\begin{assumption} \label{ass:gaussian_distribution}
	Each $d^t_{ij}$ is a Gaussian random variable with mean $\mu^t_{ij}$ and covariance $\Sigma^t_{ij} \succ 0$, i.e.\ $d^t_{ij} \sim \mathcal{N}(\mu^t_{ij},\Sigma^t_{ij})$.
\end{assumption}
Consistent with several existing trajectory planning work, we assume normally distributed uncertainties to allow the analytic reformulation of the chance constraints\footnote{See~\cite{Calafiore2006} for more general distributions with analytic reformulations.} while capturing general uncertainties with a unimodal and unbounded support set. Under Assumption~\ref{ass:gaussian_distribution}, if $\epsilon^t_{ij} < \flatfrac{1}{2}$, each chance constraint of \eqref{CCFHOC_SingleCCs} is equivalent to the following second-order cone constraint~\cite[Theorem~2.1]{Calafiore2006} for $\tilde{x} \coloneqq [x^\Trans \,\, 1]^\Trans \in \Reals^{n_x+1}$:
\begin{equation} \label{eq:CCFHOC_SingleConstr_norm}
	\Psi^{-1}(1-\epsilon^t_{ij}) \norm{(\Sigma^t_{ij})^{\flatfrac{1}{2}} \tilde{x}_t}_2 \leq {\mu^t_{ij}}^\Trans \tilde{x}_t + Mz^t_{ij} \,.
\end{equation}

Putting the above developments together we can conclude that we can conservatively satisfy the joint chance constraint \eqref{eq:CCFHOC_constr} of our original problem through $\eqref{eq:CCFHOC_constr} \Leftarrow \eqref{CCFHOC_SingleCCs} \Leftrightarrow \eqref{eq:CCFHOC_SingleConstr_norm}$. Hence, Problem~\ref{eq:CCFHOC} is conservatively reformulated as:
\vspace{-0.25\baselineskip}
\begin{mini!}<b>
	{\vb{u},z^t_{ij}}{J(x_0,\vb{u}) \label{eq:CCFHOC_2_cost}}
	{\label{eq:CCFHOC_2}}{}
	\addConstraint{\vb{x},\vb{u} \text{ satisfy \eqref{eq:x_Dynamics} with initial state } x_0}{\label{eq:CCFHOC_2_dynamics}}
	\addConstraint{\vb{u} \in \bm{\mathcal{U}}} {\label{eq:CCFHOC_2_inputs}}
	\addConstraint{\Psi^{-1}(1-\epsilon^t_{ij}) \norm*{(\Sigma^t_{ij})^{\flatfrac{1}{2}} \tilde{x}_t}_2 \leq}{\notag}
	\addConstraint{\hspace{0.37\linewidth} {\mu^t_{ij}}^\Trans \tilde{x}_t + Mz^t_{ij}}{\label{eq:CCFHOC_2_constr}}
	\addConstraint{\sum_{i=1}^{F_j} z^t_{ij} < F_j ,\, z^t_{ij} \in \{0,1\}}{\label{eq:CCFHOC_2_binvar_constr}}
\end{mini!}
\vspace{-0.6\baselineskip}\par\noindent
where constraints \eqref{eq:CCFHOC_2_constr}--\eqref{eq:CCFHOC_2_binvar_constr} must hold for all values of indices $t,j$ and $i$. Problem \eqref{eq:CCFHOC_2} is now in the standard form of a Mixed-Integer Second-Order Cone Program (MISOCP), solvable by off-the-shelf optimization solvers.

\section{MOMENTS ROBUST APPROACH} \label{sec:MRA}
The moments $\mu^t_{ij}$ and $\Sigma^t_{ij}$ capture our knowledge about the distribution. In practice one might have sample data of $d^t_{ij}$, from their distributions, and would estimate the moments from these samples. To ensure satisfaction of the chance constraints in \eqref{eq:CCFHOC_2} the uncertainty in the moments needs to be incorporated. The goal of this section is to derive tight bounds on these uncertainties and to use these bounds to derive a conservative reformulation of the chance constraints accordingly. In the following, for the sake of brevity, the indices $i$, $j$ and $t$ are omitted from $d^t_{ij}$ and its moments. Furthermore, the dimension of $d^t_{ij}$ is denoted by $n$.

\subsection{Moment concentration inequalities}
\begin{assumption} \label{ass:samples}
	We have extracted $N_s$ i.i.d.\ samples $d_1,\dots,d_{N_s}$ of $d \sim \mathcal{N}(\mu,\Sigma)$.
\end{assumption}
We form the sample mean and covariance estimates $\hat{\mu} = \flatfrac{1}{N_s} \sum_{i=1}^{N_s} d_i$ and $\hat{\Sigma} = \flatfrac{1}{(N_s-1)} \sum_{i=1}^{N_s} (d_i-\hat{\mu})(d_i-\hat{\mu})^\Trans$, which we assume to be positive definite, i.e. $\hat{\Sigma} \succ 0$\footnote{This assumption would not be reasonable for a very small number of samples $N_s < n$, since the sample covariance matrix is a sum of $N_s$ rank-1 matrices. In practice, for our trajectory generation, it is reasonable to assume that $N_s$ is sufficiently larger than $n$.}.

\begin{lemma} \label{lem:mean_bound}
	With probability $1-\beta$:
	\begin{equation} \label{eq:mean_bound}
		\norm{\mu-\hat{\mu}}_2 \leq r_1 \coloneqq \sqrt{\frac{T_{n,N_s-1}^2(1-\beta)}{N_s \lambda_{\min}(\hat{\Sigma}^{-1})}} \,,
	\end{equation}
	where $T^2_{a,b}(p)$ denotes the $p$-th quantile of the Hotelling's T-squared distribution with parameters $a$ and $b$, and $\beta \in (0,1)$.
\end{lemma}
\begin{proof}
	The sample covariance matrix of the sample mean is $\hat{\Sigma}_{\hat{\mu}} = \flatfrac{\hat{\Sigma}}{N_s}$. Furthermore, because of the underlying Gaussian assumption, the sample mean concentration:
	\begin{equation*}
		(\mu-\hat{\mu})^\Trans {\hat{\Sigma}_{\hat{\mu}}}^{-1} (\mu-\hat{\mu}) = N_s(\mu-\hat{\mu})^\Trans {\hat{\Sigma}}^{-1} (\mu-\hat{\mu}) \,,
	\end{equation*}
	follows the Hotelling's T-squared distribution~\cite{Hotelling1931} with parameters $n$ and $N_s-1$, denoted by $T^2_{n,N_s-1}$, i.e.:
	\begin{equation*}
		N_s(\mu-\hat{\mu})^\Trans {\hat{\Sigma}}^{-1} (\mu-\hat{\mu}) \sim T^2_{n,N_s-1} \,.
	\end{equation*}
	Hence, we can construct the $1-\beta$ confidence interval:
	\begin{equation*}
		N_s(\mu-\hat{\mu})^\Trans {\hat{\Sigma}}^{-1} (\mu-\hat{\mu}) \leq T^2_{n,N_s-1}(1-\beta) \,,
	\end{equation*}
	and bound the concentration of the sample mean as follows:
	\begin{IEEEeqnarray*}{ L L }
		\lambda_{\min}(\hat{\Sigma}^{-1}) \norm{\mu-\hat{\mu}}^2_2 & \leq (\mu-\hat{\mu})^\Trans {\hat{\Sigma}}^{-1} (\mu-\hat{\mu}) \\
		& \leq \frac{T^2_{n,N_s-1}(1-\beta)}{N_s} \,,
	\end{IEEEeqnarray*}
	with probability $1-\beta$, from which the statement of the lemma readily follows.
\end{proof}

\begin{lemma} \label{lem:var_bound}
	Define the constants:
	\begin{equation} \label{eq:var_diag_bound}
	\begin{split}
		r_{2,i} \coloneqq \hat{\Sigma}_{ii} \max\Big\{ 
		& \abs\big{ 1-\flatfrac{(N_s-1)}{\chi^2_{N_s-1,1-\flatfrac{\beta}{(2n)}}} } , \\
		&\abs\big{ 1-\flatfrac{(N_s-1)}{\chi^2_{N_s-1,\flatfrac{\beta}{(2n)}}} } \Big\} \,,
	\end{split}
	\end{equation}
	for all $i=1,\dots,n$, where $\chi^2_{k,p}$ is the $p$-th quantile of the $\chi^2$ distribution with $k$ degrees of freedom. Then, the inequality:
	\begin{equation} \label{eq:var_bound}
	\begin{split}
		& \norm{\Sigma-\hat{\Sigma}}_{\mathrm{F}} \leq r_2 \coloneqq \Bigg( \sum_{i=1}^{n} r^2_{2,i} + \\
		& \hspace{1.5pt} \sum_{i=1}^{n} \sum_{\substack{j=1 \\ j\neq i}}^{n} \Big( \sqrt{(\hat{\Sigma}_{ii}+r_{2,i})(\hat{\Sigma}_{jj}+r_{2,j})} + \abs*{\hat{\Sigma}_{ij}} \Big)^2 \Bigg)^{\flatfrac{1}{2}} ,
	\end{split}
	\end{equation}
	holds with probability $1-\beta$, where $\beta \in (0,1)$.
\end{lemma}

\begin{proof}
	It is known~\cite[p.~535]{Rao1965} that the diagonal elements of a sample Gaussian covariance matrix follow a scaled chi-squared distribution. Hence, we construct confidence intervals for each diagonal element separately using the sample variance concentration of a univariate Gaussian~\cite[p.~133]{Krishnamoorthy2006}. Thus, we obtain that each of the bounds:
	\begin{equation} \label{eq:var_diag_bound_r1}
		\abs*{\Sigma_{ii}-\hat{\Sigma}_{ii}} \leq r_{2,i} , \quad \forall i=1,\dots,n \,,
	\end{equation}
	holds with probability $1-\flatfrac{\beta}{n}$ and consequently all hold jointly with probability $(1-\flatfrac{\beta}{n})^n > 1-\beta$.\footnote{By considering the function $f(\beta)=(1-\flatfrac{\beta}{n})^n+\beta-1$ and verifying that $f(0)=0$ and $f'(\beta)>0$ for $n \geq 1$ and $0<\beta<1$.} Since the covariance matrix $\Sigma$ is positive definite its diagonal entries are non-negative and, moreover, it is known~\cite[p.~398]{Horn1985} that:
	\begin{equation} \label{eq:var_offdiag_bound}
		\abs*{\Sigma_{ij}} \leq \sqrt{\Sigma_{ii} \Sigma_{jj}} \leq \sqrt{(\hat{\Sigma}_{ii}+r_{2,i})(\hat{\Sigma}_{jj}+r_{2,j})} \,,
	\end{equation}
	for all $i,j=1,\dots,n$. According to the definition of the Frobenius norm we write:
	\begin{IEEEeqnarray*}{ L L }
		& \hspace{-1.5pt} \norm{\Sigma-\hat{\Sigma}}^2_{\mathrm{F}} \stackrel{\hphantom{\eqref{eq:var_offdiag_bound}}}{=} \sum_{i=1}^{n} \sum_{j=1}^{n} \abs*{\Sigma_{ij}-\hat{\Sigma}_{ij}}^2 \\
		& \hspace{-1.5pt} \stackrel{\eqref{eq:var_diag_bound_r1}}{\leq} \sum_{i=1}^{n} r^2_{2,i} + \sum_{i=1}^{n} \sum_{\substack{j=1 \\ j\neq i}}^{n} \abs*{\Sigma_{ij}-\hat{\Sigma}_{ij}}^2 \\
		& \hspace{-1.5pt} \stackrel{\hphantom{\eqref{eq:var_offdiag_bound}}}{\leq} \sum_{i=1}^{n} r^2_{2,i} + \sum_{i=1}^{n} \sum_{\substack{j=1 \\ j\neq i}}^{n} \big( \abs*{\Sigma_{ij}}+\abs*{\hat{\Sigma}_{ij}} \big)^2 \\
		& \hspace{-1.5pt} \stackrel{\eqref{eq:var_offdiag_bound}}{\leq} \sum_{i=1}^{n} r^2_{2,i} + \sum_{i=1}^{n} \sum_{\substack{j=1 \\ j\neq i}}^{n} \big( \sqrt{(\hat{\Sigma}_{ii}+r_{2,i})(\hat{\Sigma}_{jj}+r_{2,j})} + \abs*{\hat{\Sigma}_{ij}} \big)^2
	\end{IEEEeqnarray*}
	with probability $1-\beta$, from which the statement of the lemma readily follows.
\end{proof}

\begin{corollary} \label{cor:var_bound_diag}
If $\Sigma$ and $\hat{\Sigma}$ are diagonal matrices, the bound of \eqref{eq:var_bound} simplifies to:
	\begin{equation} \label{eq:var_bound_diag}
		\norm{\Sigma-\hat{\Sigma}}_{\mathrm{F}} \leq r_2 \coloneqq \sum_{i=1}^{n} r^2_{2,i} \,.
	\end{equation}
\end{corollary}
\vspace{5pt}Although the bound $r_2$ derived in Lemma~\ref{lem:var_bound} holds for a general covariance matrix $\Sigma$, Corollary~\ref{cor:var_bound_diag} provides a tighter bound in the case in which both $\Sigma$ and $\hat{\Sigma}$ are diagonal.

Concentration inequalities of sample covariance matrices have been derived before both for general~\cite{Vershynin2012} and Gaussian~\cite{Koltchinskii2017} distributions. The former result is too conservative, and both results depend on undefined constants. Hence, they are not useful in our optimization framework to follow.

\subsection{Robustifying the chance-constraints}
Using our tight bound on the distance between the true moments and the estimated ones, we use the same derivations as in~\cite[Theorem~4.1]{Calafiore2006}, albeit with different bounds, to derive a conservative reformulation of a single chance constraint as follows.

\begin{lemma} \label{lem:single_chanc_const}
	The chance constraint:
	\begin{equation*}
		\Pr(d^\Trans \tilde{x} + Mz > 0) \geq 1-\epsilon \,,
	\end{equation*}
		holds with a probability of at least $1-2\beta$, provided that:
	\begin{equation*}
		\Psi^{-1}(1-\epsilon) \norm{\left(\hat{\Sigma}+r_2 I_{n+1}\right)^{\flatfrac{1}{2}} \tilde{x}}_2 +r_1 \norm{\tilde{x}}_2 \leq {\hat{\mu}}^{\Trans} \tilde{x} + Mz \,.
	\end{equation*}
\end{lemma}

\begin{proof}
	We have already established that the chance constraint $\Pr(d^\Trans \tilde{x} + Mz > 0) \geq 1-\epsilon$ is equivalent to $ \Psi^{-1}(1-\epsilon) \norm*{\Sigma^{\flatfrac{1}{2}} \tilde{x}}_2 - {\mu}^\Trans \tilde{x} - Mz \leq 0 $. According to Lemmas~\ref{lem:mean_bound} and \ref{lem:var_bound} the bounds $\norm*{\mu-\hat{\mu}}_2 \leq r_1$ and $\norm*{\Sigma-\hat{\Sigma}}_{\mathrm{F}} \leq r_2$ each hold with probability $1-\beta$ and thus hold jointly with probability $(1-\beta)^2$. It follows that:
	\begin{IEEEeqnarray*}{ l l }
		& \Psi^{-1}(1-\epsilon) \norm*{\Sigma^{\flatfrac{1}{2}} \tilde{x}}_2 - {\mu}^\Trans \tilde{x}_t - Mz \\
		{=}\: & \Psi^{-1}(1-\epsilon) \sqrt{{\tilde{x}}^\Trans (\Sigma-\hat{\Sigma}) \tilde{x} + {\tilde{x}}^\Trans \hat{\Sigma} \tilde{x} } \\
		& - (\mu-\hat{\mu})^\Trans \tilde{x} - {\hat{\mu}}^\Trans \tilde{x} - Mz \\
		\leq & \Psi^{-1}(1-\epsilon) \sqrt{ \norm*{\Sigma-\hat{\Sigma}}_{\mathrm{F}} \norm*{\tilde{x}{\tilde{x}}^\Trans}_{\mathrm{F}} + {\tilde{x}}^\Trans \hat{\Sigma} \tilde{x} } \\
		& + \norm*{\mu-\hat{\mu}}_2 \norm*{\tilde{x}}_2 - {\hat{\mu}}^\Trans \tilde{x} - Mz \\
		\leq & \Psi^{-1}(1-\epsilon) \sqrt{ r_2 \norm*{\tilde{x}{\tilde{x}}^\Trans}_{\mathrm{F}} + {\tilde{x}}^\Trans \hat{\Sigma} \tilde{x} } \\
		& + r_1 \norm*{\tilde{x}}_2 - {\hat{\mu}}^\Trans \tilde{x} - Mz \\
		= & \Psi^{-1}(1-\epsilon) \sqrt{ {\tilde{x}}^\Trans (\hat{\Sigma}+r_2 I_{n+1}) \tilde{x} } + r_1 \norm*{\tilde{x}}_2 - {\hat{\mu}}^\Trans \tilde{x} - Mz
	\end{IEEEeqnarray*}
	with probability $(1-\beta)^2 > 1-2\beta$, from which the statement of the lemma readily follows.
\end{proof}

Using the lemma above, we formulate Problem \eqref{eq:CCFHOC_2} for the case in which the moments of the uncertainty are estimated from sample data.
\vspace{-0.3\baselineskip}
\begin{mini!}<b>
	{\vb{u}}{J(x_0,\vb{u}) \label{eq:CCFHOC_3_cost}}
	{\label{eq:CCFHOC_3}}{}
	\addConstraint{\vb{x},\vb{u} \text{ satisfy \eqref{eq:x_Dynamics} with initial state } x_0}{\label{eq:CCFHOC_3_dynamics}}
	\addConstraint{\vb{u} \in \bm{\mathcal{U}}} {\label{eq:CCFHOC_3_inputs}}
	\addConstraint{\Psi^{-1}(1-\epsilon^t_{ij}) \norm*{({\hat{\Sigma}}^t_{ij}+r^t_{2,ij}I_{n+1})^{\flatfrac{1}{2}} \tilde{x}_t}_2}{\notag}
	\addConstraint{\hspace{0.2\linewidth} + r^t_{1,ij} \norm*{\tilde{x}_t}_2 \leq {\hat{\mu}_{ij}}^{t\Trans} \tilde{x}_t + Mz^t_{ij}}{\label{eq:CCFHOC_3_constr}}
	\addConstraint{\sum_{i=1}^{F_j} z^t_{ij} < F_j ,\, z^t_{ij} \in \{0,1\}}{\label{eq:CCFHOC_3_binvar_constr}}
\end{mini!}
\vspace{-0.5\baselineskip}\par\noindent

\begin{theorem} \label{thm:CCFHOC_samples}
	A solution to Problem \eqref{eq:CCFHOC_3} is a feasible solution to Problem \eqref{eq:CCFHOC_2} with a probability of at least $ 1 - 2 \beta N \sum_{j=1}^{N_o}F_j $.
\end{theorem}

\begin{proof}
	By comparing problems \eqref{eq:CCFHOC_2} and \eqref{eq:CCFHOC_3}, it is evident that the only difference is the chance constraints \eqref{eq:CCFHOC_2_constr} and \eqref{eq:CCFHOC_3_constr}. Using Lemma~\ref{lem:single_chanc_const}, we know that each constraint of \eqref{eq:CCFHOC_3_constr} implies the corresponding chance constraint of \eqref{eq:CCFHOC_2_constr} with probability $1-2\beta$. By requiring that this implication holds jointly for all $k=N\sum_{j=1}^{N_o}F_j$ constraints, and noting that $(1-2\beta)^k > 1 - 2 \beta k$,\footnote{By considering the function $f(\beta)=(1-2\beta)^k+2k\beta -1$ and verifying that $f(0)=0$ and $f'(\beta)>0$ for $k \geq 1$ and $0<\beta<\flatfrac{1}{2}$.} the statement of the theorem readily follows.
\end{proof}
Note that, in contrast to the scenario approach~\cite{Sessa2018}, the number of samples does not affect the complexity of the resulting optimization problem, nor the probabilistic guarantees. However, a lower number of samples results in looser bounds on the estimated moments. This in turn affects the feasibility of the optimization problem. We will illustrate this aspect in the numerical simulation. 

We will illustrate the effectiveness of the concentration bounds through a simple example.
\begin{example} \label{ex:MRA}
	Consider the chance constrained problem 
	$ \min_{x \in \Reals} x$ s.t.\ $\Pr(x \geq \delta) \geq \num{0.95}$, where $\delta \sim \Gauss(0,1)$. We estimate the moments of $\delta$ using $N_s=\num{100}$ samples and solve the problem: (i) assuming that the moments estimated are correct (i.e.\ $r_1=r_2=0$) and (ii) robustifying against moment uncertainty (with $\beta=10^{-3}$). Doing the above $10^4$ times, we notice that the solution of (i) has approximately a \SI{51}{\percent} probability of violating the prescribed constraint upon drawing a new set of samples, whereas (ii) satisfies the chance constraint in all scenarios. Even if we were to increase $N_s$ significantly, e.g.\ to $N_s = 10^5$, the same behaviour is observed with (i) producing an \enquote{unsafe} solution around half of the times, whereas the solution of (ii) always satisfies the safety bound and, in fact, converges to the optimal value of $x^*=\Psi^{-1}(0.95)$ as $N_s \rightarrow \infty$. For comparison, the scenario approach would require \num{291} samples to satisfy the chance constraint with the same confidence and the distributionally robust approach of~\cite[Section 4]{Calafiore2006} could not even be applied due to the unbounded uncertainty.
\end{example}

\section{CASE STUDY} \label{sec:CaseStudy}
To demonstrate the applicability of our method, we examine an adaptation of~\cite[Case study 5.1]{Sessa2018}. In this case study\footnote{The code used for the case study is available on \href{https://github.com/Exomag/moments-robust-planning-ecc2019}{\tt moments-robust-planning-ecc2019}.}, the controlled car (henceforth called \enquote{ego car}) is driving with an initial forward velocity of $\SI{50}{\kilo\meter\per\hour}$ on the right lane of a two-way street when an adversary car, coming from the opposite direction, starts turning left into a side street.

The trajectory of position and orientation, i.e.\ the \emph{pose}, of the adversary car is uncertain and follows unicycle dynamics:
\begin{equation} \label{eq:CaseStudy_AdvDyn}
	\dot{\chi} \coloneqq
	\begin{bmatrix}
		\dot{y}_1 & 
		\dot{y}_2 & 
		\dot{\theta}
	\end{bmatrix}^\Trans
	= \begin{bmatrix}
		v \cos(\theta) &
		v \sin(\theta) &
		\omega
	\end{bmatrix}^\Trans \,,
\end{equation}
with a constant forward velocity of $v = \SI{22}{\kilo\meter\per\hour}$ and known initial state $(\SI{49}{\meter},\SI{1.75}{\meter},\ang{0})$. The angular velocity $\omega$ enters through a zero-order hold and at time $t$ is distributed as:
\begin{equation*}
	\omega_t \sim \Unif \left( \frac{\pi-0.66}{2(N+1)} , \min\left\{ \frac{\pi+0.66}{2(N+1)} , \frac{\pi}{2} - \sum_{\tau=0}^{t-1} \omega_\tau \right\} \right) \,,
\end{equation*}
measured in $\flatfrac{\si{\radian}}{T_s}$, which means that the adversary car completes a $\ang{90}$ counter-clockwise turn by the end of the planning horizon $N$.

The controlled car is modelled as a double-integrator:
\begin{equation} \label{eq:CaseStudy_EgoDyn}
	\begin{bmatrix}
		\dot{x}_1 &
		\dot{x}_2
	\end{bmatrix}^\Trans
	= \begin{bmatrix}
		x_3 &
		x_4
	\end{bmatrix}^\Trans \hspace{-2pt} , 
	\begin{bmatrix}
		\dot{x}_3 &
		\dot{x}_4
	\end{bmatrix}^\Trans
	= \begin{bmatrix}
		u_1 &
		u_2
	\end{bmatrix}^\Trans ,
\end{equation}
where the state $x \in \Reals^4$ contains the position $(x_1,x_2)$ of the car and the corresponding velocities $(x_3,x_4)$. The accelerations $(u_1,u_2)$ are the control inputs and are constrained to the set $ \mathcal{U} \coloneqq \{ u \in \Reals^2 : -3 \leq u_1 \leq 10 , \abs*{u_2} \leq 5 \} $. Similarly, the velocities $(x_3,x_4)$ are constrained to the set:
\begin{equation*}
	\mathcal{X} \coloneqq \left\{ x \in \Reals^4 : \abs{\frac{x_3}{ \SI{40}{\kilo\meter\per\hour}}-1} + \abs{\frac{x_4}{ \SI{20}{\kilo\meter\per\hour}}-1} \leq 1 \right\} \,.
\end{equation*}
The objective function is to maximize the terminal forward position of the car, encoded through the cost $J = - x_{1,N}$.

In order to deal with the uncertain pose $\chi_t$ of the adversary car, the joint chance constraint:
\begin{equation} \label{eq:CaseStudy_ChanceConstraint}
	\Pr(\bigwedge_{t=1}^{N} \bigvee_{i=1}^{4} {d^t_{i}(\chi_t)}^\Trans \begin{bmatrix} x_{1,t} \\ x_{2,t} \\ 1 \end{bmatrix} > 0) \geq 1 - \epsilon \,,
\end{equation}
is enforced. The exact form of the coefficients $d^t_{i}(\chi_t)$ is:
\begin{IEEEeqnarray*}{ L L } 
	d_1^t & = [ + \cos(\theta) , - \sin(\theta) , - \cos(\theta) y_1 + \sin(\theta) y_2 - \flatfrac{L}{2} ]^\Trans , \\
	d_2^t & = [ + \sin(\theta) , + \cos(\theta) , - \sin(\theta) y_1 - \cos(\theta) y_2 - \flatfrac{W}{2} ]^\Trans , \\
	d_3^t & = [ - \cos(\theta) , + \sin(\theta) , + \cos(\theta) y_1 - \sin(\theta) y_2 - \flatfrac{L}{2} ]^\Trans , \\
	d_4^t & = [ - \sin(\theta) , - \cos(\theta) , + \sin(\theta) y_1 + \cos(\theta) y_2 - \flatfrac{W}{2} ]^\Trans ,
\end{IEEEeqnarray*}
where $L$ and $W$ are the length and width, respectively, of the adversary car.

In the absence of prior knowledge of the adversary car's trajectory, we assume a Gaussian distribution of each coefficient $d^t_{i}(\chi_t)$ but with uncertain moments. If the actual distribution of the coefficients can be approximated sufficiently well by a Gaussian distribution, then the solution of our approach should still satisfy \eqref{eq:CaseStudy_ChanceConstraint}. The mean and covariance of each $d^t_{i}$ are calculated by sampling $N_s$ i.i.d.\ adversary car trajectories and the optimization problem is reformulated based on these estimates, as presented in Section~\ref{sec:MRA}. 

We consider three different cases for this case study. Case A estimates the moments using $N_s^A$ samples but assumes perfect knowledge of the moments (i.e.\ $r_1=r_2=0$), whereas cases B and C use $N_s^B$ and $N_s^C$ samples and robustify the constraints using the bounds $r_1$ and $r_2$. Discrete-time versions of \eqref{eq:CaseStudy_AdvDyn} and \eqref{eq:CaseStudy_EgoDyn} with sampling time $T_s = \SI{0.4}{\second}$ and a planning horizon of $N = 10$ (i.e.\ $\SI{4}{\second}$) are used. Both cars have length \SI{4.5}{\meter} and width \SI{2}{\meter}, which are taken into account for the relevant inequality constraints. We also enforce lane constraints. The joint chance constraint \eqref{eq:CaseStudy_ChanceConstraint} is imposed with $\epsilon = \num{0.05}$ and $\beta = 10^{-3}$ over the whole horizon. The number of samples used for the proposed approach is $N_s^A=N_s^B=\num{5000}$ for cases A/B and $N_s^C=\num{500}$ for case C. The resulting optimization problem is a MISOCP with $\num{60}$ continuous variables, $\num{40}$ binary variables and $\num{224}$ constraints. The MISOCP is solved for all three cases $\num{100}$ different times, with different realizations of the disturbance, taking on average $\SI{0.70}{\second}/\SI{1.07}{\second}/\SI{0.74}{\second}$ to solve for cases A/B/C respectively. All computations were carried out on an Intel i5 CPU at \SI{2.50}{\GHz} with \SI{8}{\giga\byte} of memory using YALMIP~\cite{Lofberg2004} and CPLEX~\cite{CPLEX}. For comparison, an open-loop scenario approach as in~\cite{Sessa2018} would require $\num{2963}$ samples leading to $\num{118704}$ constraints and \SI{6}{\second} solver time.

\vspace{-0.2\baselineskip}The solutions to one instance of the problem are presented in Fig.~\ref{fig:TurningTruck_Frames}. We can observe that under moment uncertainty (cases B and C) the ego car is forced to break sharply, in order to avoid the uncertain position of the adversary car. On the other hand, perfect knowledge of the moments (case A) leads to a feasible trajectory that avoids the adversary car with a narrow margin. As expected, in case B where a larger number of samples than case C is used and hence smaller moment uncertainty is present, a slightly better solution is produced (i.e.\ the ego car's final position is slightly to the right in case B compared to case C). As an example, for time step $t=5$ and constraint involving coefficient $d_4^5$, the values of the bounds are $r_1 = \num{0.0652}$ and $r_2 = \num{0.1559}$ for case B, and $r_1 = \num{0.2018}$ and $r_2 = \num{0.3436}$ for case C.

We evaluated the empirical violation probability of each case's solution through Monte Carlo simulations using $10^5$ new realizations of the adversary car's unknown pose. In Fig.~\ref{fig:TurningTruck_Trajectories} one of the $\num{100}$ solutions for the ego car is presented, along with the area covered by all the generated adversary car's trajectories. The distribution of the violation probability and the terminal position over the $\num{100}$ different instances solved is presented in Fig.~\ref{fig:TurningTruck_CostsViolsBoxplot}. As seen in Fig.~\ref{fig:TurningTruck_CostsBoxplot}, case A produced the smallest cost, with case B having the next smaller one and case C having the largest cost (i.e.\ smallest terminal position). As seen in Fig.~\ref{fig:TurningTruck_ViolsBoxplot} case A produced a violation probability of approximately $\SI{0.3}{\percent}$, well within the prescribed safety margin of $\SI{5}{\percent}$, whereas cases B and C produced $\SI{0}{\percent}$ violation probabilities. This is to be expected since, because of the increased uncertainty stemming from the moment estimates in cases B and C, the ego car chose the safest option which was to break sharply. It is noteworthy that the assumption of Gaussian uncertainty, although it was a simplifying assumption due to lack of further information, was sufficient to produce trajectories that still satisfied the risk level originally prescribed.

There are two main factors that contribute to the lower risk level exhibited by all three cases (fortunately), compared to the one prescribed. First, the trajectory planning framework present is an open-loop control scheme and is inherently more conservative than a closed-loop approach. Second, the use of Boole's inequality, although necessary for tractability, is a conservative reformulation.

\begin{figure}[ht!]
	\centering
	\setlength{\figureheight}{1.0\columnwidth}
	\setlength{\figurewidth}{0.79\columnwidth}
	\includegraphics{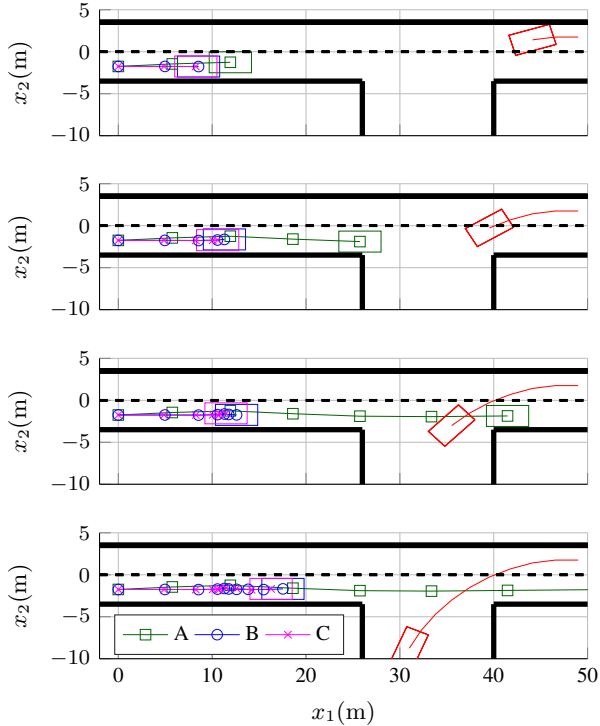}
	\caption{One simulation of the ego car's trajectories for cases A (green/squares), B (blue/circles), C (magenta/crosses) and the adversary car's trajectory (red). Pictured are time steps $t=3,5,7$ and $10$.}
	\label{fig:TurningTruck_Frames}
\end{figure}

\begin{figure}[ht!]
	\centering
	\setlength{\figureheight}{0.25\columnwidth}
	\setlength{\figurewidth}{0.79\columnwidth}
	\includegraphics{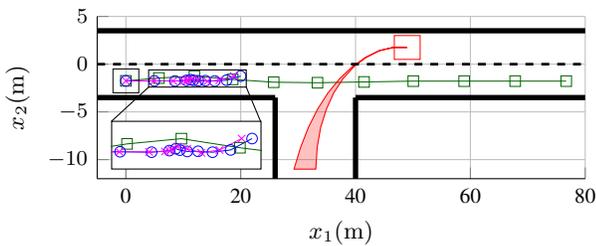}
	\caption{Evaluation of the trajectories through Monte Carlo simulations using $10^5$ realizations of adversary car trajectories. The area shaded in red contains all the new trajectories that were generated. The trajectory in green/blue/magenta is one example solution for case A/B/C, out of the $\num{100}$.}
	\label{fig:TurningTruck_Trajectories}
\end{figure}

\begin{figure}[ht!]
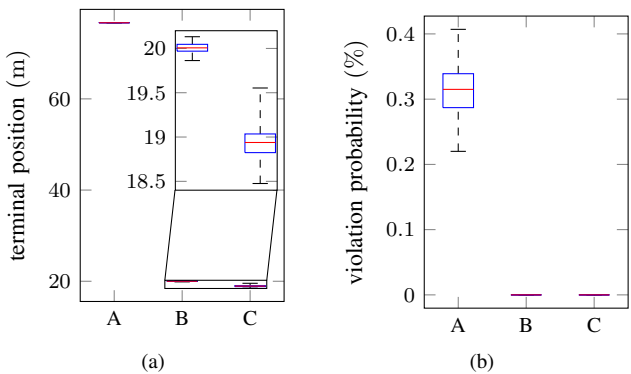

	\centering
	\setlength{\figureheight}{0.45\columnwidth}
	\setlength{\figurewidth}{0.33\columnwidth}
	\subfloat[]{
		\includegraphics{TurningTruck_CostsBoxplot.tikz}
		\label{fig:TurningTruck_CostsBoxplot}
	} \hfill
	\subfloat[]{
		\includegraphics{TurningTruck_ViolsBoxplot.tikz}
		\label{fig:TurningTruck_ViolsBoxplot}
	}
	\caption{Distribution of the (a) terminal position and (b) empirical violation probability over $\num{100}$ instances of the trajectory planning optimization problem. The median is in red and the 25/75-th percentiles are in blue.}
	\label{fig:TurningTruck_CostsViolsBoxplot}
\end{figure}

\section{CONCLUSIONS} \label{sec:Conclusions}
We used chance-constrained optimization to tackle the problem of trajectory planning in an environment comprised of uncertain obstacles. To account for partial knowledge of the uncertainty's distribution in a computationally tractable way, we reformulated the resulting problem based on the estimated moments of the distribution. We derived tight concentration inequalities for the estimated moments, in order to improve the feasibility of the resulting optimization problem. We illustrated this framework in an autonomous driving example, where the approach produced conservative, but safe, trajectories. Tuning the uncertainty's distribution model to specific problems and exploring ways to reduce the conservativeness introduced by Boole's inequality, through for example a non-uniform risk allocation, would be relevant.




\begin{thebibliography}{10}
\providecommand{\url}[1]{#1}
\csname url@rmstyle\endcsname
\providecommand{\newblock}{\relax}
\providecommand{\bibinfo}[2]{#2}
\providecommand\BIBentrySTDinterwordspacing{\spaceskip=0pt\relax}
\providecommand\BIBentryALTinterwordstretchfactor{4}
\providecommand\BIBentryALTinterwordspacing{\spaceskip=\fontdimen2\font plus
\BIBentryALTinterwordstretchfactor\fontdimen3\font minus
  \fontdimen4\font\relax}
\providecommand\BIBforeignlanguage[2]{{%
\expandafter\ifx\csname l@#1\endcsname\relax
\typeout{** WARNING: IEEEtran.bst: No hyphenation pattern has been}%
\typeout{** loaded for the language `#1'. Using the pattern for}%
\typeout{** the default language instead.}%
\else
\language=\csname l@#1\endcsname
\fi
#2}}

\bibitem{Schouwenaars2001}
T.~Schouwenaars, B.~D. Moor, E.~Feron, and J.~P. How, ``Mixed integer
  programming for multi-vehicle path planning,'' in \emph{Proc. of the 2001
  European Control Conference ({ECC})}, 2001, pp. 2603--2608.

\bibitem{Richards2002}
A.~Richards and J.~P. How, ``Aircraft trajectory planning with collision
  avoidance using mixed integer linear programming,'' in \emph{Proc. of the
  2002 American Control Conference ({ACC})}, 2002, pp. 1936--1941.

\bibitem{Blackmore2011}
L.~Blackmore, M.~Ono, and B.~C. Williams, ``Chance-constrained optimal path
  planning with obstacles,'' \emph{{IEEE} Trans. Robot.}, vol.~27, no.~6, pp.
  1080--1094, 2011.

\bibitem{Jha2018}
S.~Jha, V.~Raman, D.~Sadigh, and S.~A. Seshia, ``Safe autonomy under perception
  uncertainty using chance-constrained temporal logic,'' \emph{Journal of
  Automated Reasoning}, vol.~60, no.~1, pp. 43--62, 2018.

\bibitem{Vitus2011}
M.~P. Vitus and C.~J. Tomlin, ``On feedback design and risk allocation in
  chance constrained control,'' in \emph{Proc. of the 50th {IEEE} Conference on
  Decision and Control and European Control Conference}, 2011, pp. 734--739.

\bibitem{Vitus2012}
M.~P. Vitus, ``Stochastic control via chance constrained optimization and its
  application to unmanned aerial vehicles,'' Ph.D. dissertation, Stanford
  University, 2012.

\bibitem{Paulson2017}
J.~A. Paulson, E.~A. Buehler, R.~D. Braatz, and A.~Mesbah, ``Stochastic model
  predictive control with joint chance constraints,'' \emph{International
  Journal of Control}, pp. 1--14, 2017.

\bibitem{Calafiore2006}
G.~C. Calafiore and L.~E. Ghaoui, ``On distributionally robust
  chance-constrained linear programs,'' \emph{Journal of Optimization Theory
  and Applications}, vol. 130, no.~1, pp. 1--22, 2006.

\bibitem{Delage2010}
E.~Delage and Y.~Ye, ``Distributionally robust optimization under moment
  uncertainty with application to data-driven problems,'' \emph{Operations
  Research}, vol.~58, no.~3, pp. 595--612, 2010.

\bibitem{Sessa2018}
P.~G. Sessa, D.~Frick, T.~A. Wood, and M.~Kamgarpour, ``From uncertainty data
  to robust policies for temporal logic planning,'' in \emph{Proc. of the 21st
  International Conference on Hybrid Systems: Computation and Control}, 2018,
  pp. 157--166.

\bibitem{Casella2002}
G.~Casella and R.~L. Berger, \emph{Statistical Inference}.\hskip 1em plus 0.5em
  minus 0.4em\relax Cengage Learning, 2001.

\bibitem{Hotelling1931}
H.~Hotelling, ``The generalization of student's ratio,'' \emph{The Annals of
  Mathematical Statistics}, vol.~2, no.~3, pp. 360--378, 1931.

\bibitem{Rao1965}
C.~R. Rao, \emph{Linear Statistical Inference and its Applications}.\hskip 1em
  plus 0.5em minus 0.4em\relax John Wiley and Sons, Inc., 1965.

\bibitem{Krishnamoorthy2006}
K.~Krishnamoorthy, \emph{Handbook of Statistical Distributions with
  Applications}.\hskip 1em plus 0.5em minus 0.4em\relax Chapman and Hall/CRC,
  2006.

\bibitem{Horn1985}
R.~A. Horn and C.~R. Johnson, Eds., \emph{Matrix Analysis}.\hskip 1em plus
  0.5em minus 0.4em\relax Cambridge University Press, 1986.

\bibitem{Vershynin2012}
R.~Vershynin, ``How close is the sample covariance matrix to the actual
  covariance matrix?'' \emph{Journal of Theoretical Probability}, vol.~25,
  no.~3, pp. 655--686, 2012.

\bibitem{Koltchinskii2017}
V.~Koltchinskii and K.~Lounici, ``Concentration inequalities and moment bounds
  for sample covariance operators,'' \emph{Bernoulli}, vol.~23, no.~1, pp.
  110--133, 2017.

\bibitem{Lofberg2004}
J.~Lofberg, ``{YALMIP} : a toolbox for modeling and optimization in {MATLAB},''
  in \emph{Proc. of the 2004 {IEEE} International Conference on Robotics and
  Automation}, 2004, pp. 284--289.

\bibitem{CPLEX}
{International Business Machines Corporation (IBM)}, ``{IBM ILOG CPLEX
  Optimization Studio},'' 2017.

\end{thebibliography}
\end{document}